\documentclass[11pt]{article}
\usepackage[letterpaper,margin=1.05in]{geometry}
\usepackage[T1]{fontenc}
\usepackage{lmodern,microtype}
\usepackage{amsmath,amssymb,amsthm,mathtools}
\usepackage{enumitem,booktabs}
\usepackage{xcolor}
\usepackage{hyperref}
\hypersetup{colorlinks=true,linkcolor=blue!45!black,citecolor=blue!45!black,urlcolor=blue!45!black}
\usepackage{setspace}
\allowdisplaybreaks
\numberwithin{equation}{section}
\newtheorem{theorem}{Theorem}
\newtheorem{lemma}[theorem]{Lemma}
\newtheorem{proposition}[theorem]{Proposition}
\newtheorem{corollary}[theorem]{Corollary}
\theoremstyle{definition}

\theoremstyle{remark}

\newcommand{\PP}{\mathbb P}
\newcommand{\EE}{\mathbb E}
\newcommand{\one}{\mathbf 1}
\newcommand{\cL}{\mathcal L}
\newcommand{\cP}{\mathcal P}
\newcommand{\TV}{\mathrm{TV}}
\newcommand{\essinf}{\operatorname*{ess\,inf}}
\newcommand{\esssup}{\operatorname*{ess\,sup}}
\newcommand{\supp}{\operatorname{supp}}
\newcommand{\Proj}{\Pi}

\newcommand{\eps}{\varepsilon}
\title{Social Learning with Selective Sampling}
\author{Zihan Zhao\thanks{Renmin University of China. zihanzhao.econtheory@gmail.com. Acknowledgment to be added.}}
\date{November 21, 2026}

\begin{document}
\maketitle
\begin{abstract}
This paper studies social learning when some types of actions are more likely to be sampled by successors than others. We show that Bayesian agents can achieve learning even without expanding observations, because the endogenous observation network itself carries information and agents can correct for selection bias. For nonanonymous sampling, we establish sufficient and pairwise necessary conditions of the sampling procedure for learning, which coincide in the binary state-matching model to yield an exact characterization. We also show that stationary anonymous sampling permits asymptotic learning, as agents can correct for selection bias using inverse-probability weighting. Finally, we demonstrate that asymptotic learning can fail under arbitrarily small misspecification of the sampling procedure.
\end{abstract}

\section{Introduction}

Social learning often takes place through selectively observed records of others' decisions. Online platforms may display purchases more often than nonpurchases, successful projects more often than failed ones, or favorable reviews more often than unfavorable ones. Empirical studies of online reviews document systematic selection in whose experiences become visible and how such selection can distort the information available to subsequent consumers (Li and Hitt, 2008; Hu, Pavlou, and Zhang, 2017). The concern is particularly important in social learning, where individuals rely on others' observed decisions to infer an unknown state of the world. When some actions are systematically more likely to be observed than others, the observed sample need not represent the underlying population of decisions. This raises a natural question: under what conditions can Bayesian agents achieve asymptotic learning despite selective observation?

The central methodological difficulty is that selective sampling makes the observation network endogenous, and more importantly dependent on the realized state. Whether a predecessor's action is observed depends on the action she takes and, through her information and decision, on the underlying state. This dependence prevents a direct application of the improvement principle of Acemoglu et al. (2011), a standard approach to social learning under incomplete observation. The principle operates through two steps: a later agent imitates an observed predecessor and then improves upon that decision using her own private signal. Under action-dependent selection, however, observing a predecessor is itself a selection event, so her conditional performance need not equal her unconditional performance. Imitating an observed predecessor therefore does not generally deliver the welfare inheritance required by the standard argument. Yet this endogeneity is also a blessing: precisely because observation depends on the action taken, the observation network itself may carry information about predecessors' decisions.

Our approach generalizes the improvement principle by modifying its first step. Rather than directly imitating an observed predecessor, we construct a feasible decision rule that “corrects for selection before imitation”, while retaining the second step of improvement through private information. Under nonanonymous observation, agents know which predecessors' actions are displayed and which are missing. The missing observations themselves are informative: a predecessor who is unlikely to be displayed when taking a particular action becomes more likely to have taken that action conditional on non-display. By exploiting this information and correcting for differences in display probabilities, we construct a randomized decision that approximately replicates a mixture of sufficiently late predecessors' actions. This restores welfare inheritance even when the realized observation network is not expanding and the sampling technology is nonstationary.

We also extend the argument to anonymous observation, where agents observe only aggregate counts of displayed actions. Although individual predecessors can no longer be identified, under stationary sampling the selection bias can be corrected through a procedure analogous to inverse-probability weighting. By downweighting disproportionately displayed actions, and recovering the frequency of an unobserved action from the adding-up restriction when necessary, agents can consistently estimate the empirical distribution of predecessors' actions. Randomizing according to this estimate approximately replicates a uniformly selected predecessor, again restoring welfare inheritance. In both observation structures, the corrected imitation step, combined with the value of a fresh private signal, delivers asymptotic learning.

Our main results characterize when selective observation permits asymptotic learning. For general finite state and action spaces, we establish “Generalized Expanding Detectability (GED)” as a sufficient condition for adequate learning under appropriate preference–information assumptions, and “Pairwise Expanding Detectability (PED)” as a necessary condition across payoff-reversing binary choice sets. In the binary state-matching model with two-sided unbounded private beliefs, these conditions coincide, yielding an exact necessary-and-sufficient characterization of asymptotic learning.

As an extension, we show that this learning result need not be robust to misspecification: arbitrarily small errors in the perceived sampling technology can generate persistent learning failures, even when both the true and perceived observation kernels satisfy our generalized exposure condition.

\subsection{Related Literature}

The social learning literature has developed two main approaches to asymptotic learning. Under complete observation, as in Smith and Sørensen (2000), nested information permits martingale arguments. Under incomplete observation, where information sets are generally non-nested, a prominent alternative is the “improvement principle” of Acemoglu et al. (2011). Kartik et al. (2024) extend this approach beyond unbounded private beliefs, using welfare inheritance and private-information improvement to establish adequate learning under excludability. Our proof follows this second approach. Since action-dependent selection makes direct imitation biased, we construct a corrected-imitation rule that approximately replicates sufficiently late predecessors' actions, restoring welfare inheritance.

Our generalized expanding detectability (GED) relates to several network expansion conditions: expanding observations (EO) in Acemoglu et al. (2011), expanding subnetworks (ESN) in Lobel and Sadler (2015), and expanding sample sizes (ESS) and expanding nested samples (ENS) in Cremin (2024). In the conventional network framework, these satisfy \(\mathrm{ENS}\Rightarrow\mathrm{ESS}\Rightarrow\mathrm{ESN}\Rightarrow\mathrm{EO}\). GED coincides with EO under action-independent sampling, and with ESN when neighborhoods are independently generated. Under action-dependent observation, however, displayed edges need not capture all available information: non-display may itself reveal a predecessor's action. In the binary model, realized-network ESN implies our expanding detectability (ED), but not conversely. For general finite action sets, the conditions need not be ordered. GED thus captures “recoverability” rather than connectivity alone, providing a sufficient condition for learning, complemented by pairwise necessity conditions and an exact binary characterization.

Herrera and Hörner (2013) and Acemoglu et al. (2022) study closely related economic forces: agents' decisions endogenously determine which observations become available to subsequent decision makers. Herrera and Hörner (2013) consider a Poisson arrival process in which only investment decisions are publicly observed, so the absence of investment may reflect either noninvestment or the absence of an arrival. Acemoglu et al. (2022) study learning from online reviews, where purchasing decisions determine which consumers generate reviews. Under full-history disclosure, the common public history allows consumers to reconstruct previous beliefs and account for selection in Bayesian inference; under summary-statistic disclosure, they establish learning through conditions separating the distributions of reviews across states.

Our paper differs in the structure of endogenous observation and, consequently, in how selection is overcome. Rather than a common public history generated by a particular arrival or review process, we consider a general, potentially nonstationary observation network in which each predecessor is independently sampled by each successor with a probability depending on the predecessor's action. Consequently, different agents receive different selected subsets of past decisions, and the complete history needed to reconstruct the beliefs underlying those decisions is generally unavailable. Our approach is therefore not to recover the full public belief through a likelihood calculation, but to identify conditions under which agents can correct for selection in their own observations and recover the decision-making performance of earlier agents. This corrected imitation restores the welfare-inheritance step of the improvement principle, allowing private information to deliver further improvement and ultimately asymptotic learning.

Our anonymous extension is closely related to Smith and Sørensen (2020), who study unordered random samples of predecessors' actions and use welfare-improvement arguments, together with Pólya-urn methods for recursive sampling. Their sampling is action-independent, allowing imitation of a randomly sampled predecessor. Under our action-dependent sampling, this imitation is biased. With stationary selection, however, inverse-probability weighting of anonymous counts consistently estimates predecessors' empirical action distribution, permitting approximate imitation of a uniformly selected predecessor. We therefore recover welfare inheritance without characterizing action-frequency dynamics, at the cost of restricting attention to stationary, independent fresh sampling.

Fudenberg, Lanzani, and Strack (2024) study selective-memory equilibrium in a single-agent environment, where agents naively treat selectively recalled experiences as their complete histories. Their long-run analysis characterizes convergent behavior and provides conditions for convergence. Our setting differs because predecessor actions are endogenous Bayesian decisions: their order and informational context matter, so action frequencies alone are generally insufficient to summarize a social history. More importantly, our agents correctly understand selection. In a single-agent environment, awareness of the selection procedure would permit inverse-probability correction and, under suitable identification and sampling conditions, consistent learning. Our results instead establish how such correction supports learning when the selected observations are themselves socially generated.

Finally, our misspecification extension relates to Frick, Iijima, and Ishii (2020, 2023), Bohren and Hauser (2021), and Liu (2025). Frick et al. (2020) establish fragility under small misperceptions about others' types, while their 2023 paper develops martingale methods for misspecified learning. Bohren and Hauser characterize learning under heterogeneous misspecified models, and Liu studies nonuniform robustness. Our agents instead misperceive the observation kernel, and fresh sampling generates non-nested histories. We identify a distinct mechanism: a missing display under the true kernel is interpreted as a definite predecessor action under the subjective kernel, creating persistent errors. Arbitrarily small kernel perturbations can therefore destroy learning even when both kernels satisfy GED and each supports learning under correct specification. This establishes a failure of uniform robustness within the GED class, without ruling out local robustness at every fixed kernel.

\section{setup}
\subsection{Preferences, signals, and equilibrium}
Let $\Omega$ be a finite state space and $A$ a finite action set. A single state $\omega\in\Omega$ is drawn at date zero from a common prior $\pi\in\Delta(\Omega)$. Agents $t=1,2,\dots$ move sequentially and share a bounded utility function $u:A\times\Omega\to\mathbb R$. Each agent chooses only her own action; there are no direct payoff externalities or incentives to signal to future agents. We sometimes restrict all agents to a common nonempty choice set $B\subseteq A$, with the same $u$ and information technology.

Agent $t$ receives a private signal $S_t$ in a
standard Borel space $\mathsf S$. Conditional on
$\omega=\theta$, the private signals
$S_1,S_2,\ldots$ are independent and identically
distributed with law $F_\theta$. They are
independent of all sampling and exogenous decision
randomizations conditional on $\omega$.

We assume that preferences and private information
satisfy excludability in the sense of
Kartik et al (2024). For any two
actions $a,b\in A$, let
\[
\Omega_{a,b}
=
\{\omega\in\Omega:u(a,\theta)>u(b,\theta)\}.
\]
Excludability requires that $\Omega_{a,b}$ be
distinguishable from $\Omega_{b,a}$: for every
$\varepsilon>0$ and every belief
$\mu\in\Delta(\Omega_{a,b}\cup\Omega_{b,a})$
with $\mu(\Omega_{a,b})>0$,
\[
\mathbb P_\mu\left(
  \mu(\Omega_{a,b}\mid S_t)>1-\varepsilon
\right)>0.
\]

We impose this assumption to abstract from the
classical failures of asymptotic learning caused
by informationally induced herding, which are
not the focus of this paper. Our interest is
instead in whether the sampling mechanism itself
can prevent asymptotic learning, even when
private information and preferences satisfy
excludability.

A Bayesian equilibrium for menu $B$ is a sequence of measurable decision rules that select a best response in $B$ given each agent's social observation (specified in the next subsection) and private signal;  We consider arbitrary equilibria, and our results do not require an equilibrium selection. Note that across the paper, agents always have correct knowledge of their own indices. 

\subsection{Selective Sampling}
For each pair of agents $i<t$, let $D_{i,t}\in\{0,1\}$ indicate
whether agent $i$'s action is displayed to agent $t$. The sampling
technology is a commonly known collection of functions
\[
Q_{i,t}:A\longrightarrow[0,1],\qquad i<t.
\]
At date $t$, conditional on the realized predecessor actions
$A_{<t}=(A_1,\ldots,A_{t-1})$, the indicators $(D_{i,t})_{i<t}$
are drawn independently, with
\begin{equation}\label{eq:sampling-kernel}
\mathbb P(D_{i,t}=1\mid A_{<t},\omega)=Q_{i,t}(A_i).
\end{equation}
Thus some types of actions are more likely to be sampled by successors than others, which we call "selective". The observation network formalized by such a procedure is random. The randomness itself doesn't prevent applying the improvement principle, but selection does: Acemoglu et al (2011), Lobel and Sadler (2015) and Kartik et al (2024) all rely on the independence of network topology and the realized state; In contrast,
our sampling procedure makes the realized network endogenous:
whether a predecessor is observed depends on her action, which
in turn depends on her information about the state. So imitating a sampled predecessor
does not generally preserve her ex ante expected payoff.
Consequently, the welfare-inheritance step of the conventional
improvement principle no longer follows directly. At the same
time, the endogeneity means that
the network can itself convey information about predecessors'
decisions

\section{Main Results}
\subsection{Nonanonymous Sampling}
Our main results concern nonanonymous observation with the possibly
nonstationary, action-dependent sampling kernels introduced in
Section 2. Fix a nonempty menu $B\subseteq A$.
Write $U_t^B=\mathbb E[u(A_t,\omega)]$ for the equilibrium payoff when all
agents choose from $B$, and
\[
U_B^*:=\sum_{\omega\in\Omega}\pi(\omega)\max_{a\in B}u(a,\omega)
\]
for the corresponding full-information benchmark. We say that
\emph{adequate learning} obtains at $B$ if $U_t^B\to U_B^*$.

\subsubsection{Sufficient Conditions for Adequate Learning}
\label{subsec:GED}

For $|B|\ge2$, define the \emph{pairwise distinguishability} of a
predecessor's sampling record by
\begin{equation}\label{eq:rho}
\rho_{i,t}(B):=\min_{\substack{a,b\in B\\a\ne b}}
                  \max\{Q_{i,t}(a),Q_{i,t}(b)\}.
\end{equation}
Fix a minimizer $b_{i,t}\in\arg\min_{a\in B}Q_{i,t}(a)$ (ties broken
arbitrarily). Then
\begin{equation}\label{eq:second}
Q_{i,t}(a)\ge\rho_{i,t}(B)\qquad (a\in B\setminus\{b_{i,t}\}).
\end{equation}
The quantity $\rho_{i,t}(B)$ is the second-smallest display probability
across actions in $B$. Its relevance is that a non-display can identify
the least-displayed action: distinguishing $a$ from $b$ only requires
one of them to be displayed with sufficiently high probability. If all other actions are displayed, non-display reveals the remaining action. Hence, at most one action can be left undetected without compromising identification. The relevant measure of effective exposure is therefore the second-smallest display probability, rather than the smallest.

We call the following condition \emph{generalized expanding
detectability} (GED):
\begin{equation}\label{eq:GED}
\mathrm{GED}(B):\qquad
\text{for every fixed }K\ge0,\qquad
\prod_{i=K+1}^{t-1}\bigl(1-\rho_{i,t}(B)\bigr)
\xrightarrow[t\to\infty]{}0.
\end{equation}

\begin{theorem}[Adequate learning under nonanonymous sampling]
\label{thm:GED}
suppose that
$\mathrm{GED}(B)$ holds,
Then every Bayesian equilibrium of the game equipped with $B$ satisfies
\begin{equation}\label{eq:adequate}
\lim_{t\to\infty}\mathbb E[u(A_t,\omega)]=U_B^*.
\end{equation}
\end{theorem}

The mechanism behind Theorem~\ref{thm:GED} is \emph{corrected imitation}.
Under ordinary exogenous observation, an agent can imitate a displayed
predecessor and inherit that predecessor's expected payoff. With selective
sampling, displayed actions are a biased sample and this comparison
need not hold. Nonanonymous observation nevertheless supplies both the
identity of each predecessor and the event that her action was not
displayed. Together with the known kernel $Q_{i,t}$, these observations
allow a later agent to remove the selection bias before imitating.

For each $i<t$, let
\[
X_{i,t}=
\begin{cases}
A_i,&D_{i,t}=1,\\
\varnothing,&D_{i,t}=0.
\end{cases}
\]
Under nonanonymous observation, agent $t$
observes the identity-indexed record
$O_t=(X_{i,t})_{i<t}$.

To see how the correction works, consider a predecessor $i<t$, let
$r_i=\rho_{i,t}(B)$, and choose
$b_i\in\arg\min_{a\in B}Q_{i,t}(a)$. By construction,
$Q_{i,t}(a)\ge r_i$ for all $a\ne b_i$. For $r_i>0$, whenever
$X_{i,t}=a\ne b_i$, independently accept the displayed action with
probability $r_i/Q_{i,t}(a)$, and denote this event by $H_{i,t}=1$.
With $e_a$ the unit probability vector at $a$, form the observable
vector statistic
\begin{equation}\label{eq:main-Z}
Z_{i,t}:=r_i e_{b_i}
 +\sum_{a\in B\setminus\{b_i\}}\mathbf1\{X_{i,t}=a,\,H_{i,t}=1\}(e_a-e_{b_i}),
\qquad (r_i>0),
\end{equation}
where the indicator's action $a$ is the displayed value; set
$Z_{i,t}=0$ when $r_i=0$. For every possible realized action $A_i$,
\begin{equation}\label{eq:main-Zmean}
\mathbb E[Z_{i,t}\mid A_{<t},\omega]=r_i e_{A_i},\qquad
\mathbb E\bigl[\|Z_{i,t}-r_i e_{A_i}\|_2^2\mid A_{<t},\omega\bigr]
\le2r_i.
\end{equation}
Thus the least-displayed action is recovered as a baseline and the
other actions are thinned to the same effective rate. Missingness is
used as information rather than treated as an observation to be
discarded. Importantly, the right-hand side of the first identity
depends on the \emph{unconditional realized action}, not its distribution
conditional on being displayed.

GED ensures that sufficiently late predecessors can be exploited in one
of two ways. If the sum $R_t$ of their effective rates $r_i$ diverges,
independence of fresh sampling draws makes an aggregate of
\eqref{eq:main-Z} consistently recover a deterministic weighted mixture
of their realized actions. If $R_t$ does not diverge, the vanishing
product in \eqref{eq:GED} forces some sufficiently late predecessor to
have $r_i$ arbitrarily close to one: her action can then be decoded
almost perfectly from either a display or a non-display. These two
constructions provide the same decision-theoretic guarantee.

\begin{lemma}[Uniform corrected imitation]\label{lem:replica}
If $\mathrm{GED}(B)$ holds, there are deterministic weights
$w_{i,t}\ge0$ with $\sum_{i<t}w_{i,t}=1$, and feasible social-only
randomized decisions $d_t\in B$, such that
\begin{align}
\label{eq:escape}
&\sum_{i\le K}w_{i,t}\longrightarrow0
&&\text{for every fixed }K,\\
\label{eq:replica}
&\sup_{a_{<t}\in B^{t-1}}
\left\|\mathcal L(d_t\mid A_{<t}=a_{<t})
 -\sum_{i<t}w_{i,t}\delta_{a_i}\right\|_{\rm TV}
\longrightarrow0.
\end{align}
The weights and the rule depend only on the sampling technology, not
on the underlying state, the payoff function, the prior, or the equilibrium.
\end{lemma}

The uniformity in \eqref{eq:replica} matters. It means that correcting
selection does not rely on equilibrium action frequencies converging
or on a conjecture about the actual distribution of predecessors'
actions. The rule works conditional on \emph{every} action history,
including histories that have zero equilibrium probability.

To connect this statistical construction to welfare, for a belief
$\mu\in\Delta(\Omega)$ define
\begin{align}
 v_B(\mu)&=\max_{a\in B}\sum_{\omega}\mu(\omega)u(a,\omega),
 \label{eq:v}\\
 v_B^F(\mu)&=\int_{\mathsf S}\max_{a\in B}
    \sum_{\omega}\mu(\omega)f_\omega(s)u(a,\omega)\,\nu(ds),
 \label{eq:vs}\\
 v_B^*(\mu)&=\sum_{\omega}\mu(\omega)
                   \max_{a\in B}u(a,\omega),
 \label{eq:vstar}
\end{align}
where $\nu=\sum_{\omega\in\Omega}F_\omega$ and
$f_\omega=dF_\omega/d\nu$. Put
\begin{equation}\label{eq:IG}
I_B(\mu)=v_B^F(\mu)-v_B(\mu),\qquad
G_B(\mu)=v_B^*(\mu)-v_B(\mu).
\end{equation}
Here $I_B$ is the value of one fresh private signal, while $G_B$ is
the social posterior's distance from the full-information payoff.
The relevant preference--information property is
\begin{equation}\label{eq:AK}
\mathrm{AK}(B):\qquad I_B(\mu)=0\ \Longrightarrow\ G_B(\mu)=0
\quad\text{for all }\mu\in\Delta(\Omega).
\end{equation}
It rules out an inadequately informed belief at which the next
private signal has no decision value. The excludability assumption
in the Setup, following Kartik et al. (2024), implies this property.
The theorem in fact needs only $\mathrm{AK}(B)$, not the stronger
primitive assumption.

Let $\mu_t=\mathbb P(\omega\in\cdot\mid O_t)$ denote the
social posterior before agent $t$ observes her private signal, and
write $B_t=\mathbb E[v_B(\mu_t)]$ and
$J_t=\mathbb E[I_B(\mu_t)]$.
Since $S_t$ is independent of the social record conditional on the state,
Bayesian optimality gives $U_t^B=B_t+J_t$.
The corrected-imitation rule is a feasible alternative to the
social-only optimal action. Consequently it yields the following
form of welfare inheritance.

\begin{proposition}[Corrected welfare inheritance]\label{prop:WI}
Under $\mathrm{GED}(B)$, the weights in
Lemma~\ref{lem:replica} and a deterministic sequence
$\eps_t\to0$ can be chosen so that
\begin{equation}\label{eq:WI}
 B_t\ge\sum_{i<t}w_{i,t}U_i^B-C_B\eps_t,
 \qquad
 C_B:=\max_{\omega\in\Omega}
 \bigl(\max_{a\in B}u(a,\omega)-\min_{a\in B}u(a,\omega)\bigr).
\end{equation}
In particular,
\begin{equation}\label{eq:liminf}
\liminf_{t\to\infty}B_t\ge\liminf_{t\to\infty}U_t^B.
\end{equation}
\end{proposition}

Equation~\eqref{eq:WI} replaces the ordinary imitation step of the
improvement principle. The second step is unchanged: by
\eqref{eq:AK}, inadequately informed social posteriors admit a strict
private-signal improvement. Compactness and continuity, as in the
argument of Kartik et al. (2024), make that gain uniform whenever
expected social welfare remains a fixed distance below $U_B^*$.
The two steps together preclude an asymptotic welfare loss.

\begin{lemma}[Uniform private-signal improvement]\label{lem:gain}
Assume $\mathrm{AK}(B)$ and fix the prior $\pi$. For each $\delta>0$
for which the antecedent is feasible, there exists $\eta(\delta)>0$
such that any distribution $\phi$ of beliefs with barycenter $\pi$
satisfies
\begin{equation}\label{eq:gain}
\int v_B(\mu)\,\phi(d\mu)\le U_B^*-\delta
\quad\Longrightarrow\quad
\int I_B(\mu)\,\phi(d\mu)\ge\eta(\delta).
\end{equation}
\end{lemma}

\subsubsection{Necessary Condition for Binary Decision Problems}
\label{subsec:PED}

For two distinct actions $a,b$, define
$m_{i,t}^{a,b}:=\max\{Q_{i,t}(a),Q_{i,t}(b)\}$.
We say that \emph{pairwise expanding detectability} (PED) holds for
$(a,b)$ if
\begin{equation}\label{eq:PED}
\mathrm{PED}(a,b):\qquad
\text{for every fixed }K\ge0,\qquad
\prod_{i=K+1}^{t-1}(1-m_{i,t}^{a,b})\longrightarrow0.
\end{equation}
For a binary menu $B=\{a,b\}$, GED and PED coincide. For larger menus,
GED implies PED for every pair, but the reverse implication need not
hold: the action pair minimizing $m_{i,t}^{a,b}$ can vary across
predecessors, and uniformly correcting all actions is a stronger
requirement than correcting any fixed pair.

Call $(a,b)$ \emph{strictly payoff-reversing} if there are states
$\omega_a,\omega_b$ such that
\begin{equation}\label{eq:reversal}
\Delta_{a,b}:=\min\{u(a,\omega_a)-u(b,\omega_a),
 u(b,\omega_b)-u(a,\omega_b)\}>0.
\end{equation}
Let $\beta_{a,b}:=(F_{\omega_a}\wedge F_{\omega_b})(\mathsf S)$
be the total common mass of these two states' private-signal laws.
Thus $\beta_{a,b}>0$ requires positive signal overlap, not mutual
absolute continuity.

\begin{theorem}[PED necessity for binary menus]\label{thm:PED}
Consider the restricted-menu game $B=\{a,b\}$ and a strictly
payoff-reversing pair with $\beta_{a,b}>0$. Suppose
$\pi(\omega_a),\pi(\omega_b)>0$. If $\mathrm{PED}(a,b)$ fails,
adequate learning fails in every Bayesian equilibrium.
More precisely, if along dates $t_j\to\infty$ there exist a fixed
$K\ge0$ and $\alpha>0$ such that
\begin{equation}\label{eq:failPED}
 \prod_{i=K+1}^{t_j-1}(1-m_{i,t_j}^{a,b})\ge\alpha,
\end{equation}
then every equilibrium satisfies
\begin{equation}\label{eq:loss}
 U_B^*-U_{t_j}^B\ge
 \Delta_{a,b}\,\alpha\,\min\{\pi(\omega_a),\pi(\omega_b)\}
 \beta_{a,b}^{K+1}>0.
\end{equation}
Consequently, adequate learning for every payoff-reversing binary
menu with overlapping witness-state signals requires PED for each
such pair.
\end{theorem}

The obstruction is \emph{indistinguishability}, not an assumption that
non-displays are uninformative. When \eqref{eq:failPED} holds, the
probability of seeing no predecessor beyond a fixed prefix is at
least $\alpha$ under \emph{every} complete binary action history.
The laws of a finite prefix of actions have a common component under
$\omega_a$ and $\omega_b$ because the first $K$ private signals
have overlapping distributions. Passing this common component
through the prefix sampling technology, then adjoining the current
private signal, leaves a positive common component in the two
state-conditioned information experiments. Their overlap imposes the
welfare-loss lower bound \eqref{eq:loss}, independently of equilibrium
strategies.

The binary-menu qualification is essential. At a fixed larger menu,
a predecessor may take an action $c$ with much higher display
probability than either $a$ or $b$; other actions may also be
irrelevant to optimal decisions. Thus the theorem does \emph{not}
claim that PED for every payoff-reversing pair is necessary at a
single fixed multi-action menu.

\subsubsection{A sharp result for binary state matching}\label{sec:binary}

Let $\Omega=A=\{0,1\}$, $u(a,\omega)=\mathbf1\{a=\omega\}$,
and $\pi(0),\pi(1)>0$. Assume that $F_0$ and $F_1$ are mutually
absolutely continuous and that their likelihood ratio
$L=dF_1/dF_0$ satisfies two-sided unboundedness,
\begin{equation}\label{eq:UB}
\essinf_{F_0}L=0,\qquad \esssup_{F_0}L=+\infty.
\end{equation}
Here $\rho_{i,t}(A)=m_{i,t}^{0,1}$, so GED and PED reduce to one
condition, \emph{expanding detectability} (ED):
\begin{equation}\label{eq:ED}
\mathrm{ED}:\qquad \forall K\ge0,\quad
\prod_{i=K+1}^{t-1}\left(1-\max\{Q_{i,t}(0),Q_{i,t}(1)\}\right)
\longrightarrow0.
\end{equation}

\begin{corollary}[Binary ED if and only if learning]\label{cor:binary}
Under the binary assumptions above, the following are equivalent:
\begin{enumerate}[label=(\roman*)]
\item $\mathrm{ED}$ holds.
\item Every Bayesian equilibrium satisfies
$\mathbb P(A_t=\omega)\to1$.
\end{enumerate}
If ED fails, learning fails in every Bayesian equilibrium.
\end{corollary}

The characterization shows why it is misleading to assess a
selectively observed network only by its displayed edges. For
example, suppose only action $1$ can ever be displayed, with
$Q_{t-1,t}(1)=1$ and $Q_{t-1,t}(0)=0$, and no other predecessor
is displayed. The agent recovers $A_{t-1}$ exactly, since a
non-display identifies $0$. Thus ED holds although at most one
predecessor is ever displayed. With action-independent selection,
by contrast, $Q_{i,t}(0)=Q_{i,t}(1)=q_{i,t}$ and ED reduces to the
standard expanding-observation probability
$\prod_{i=K+1}^{t-1}(1-q_{i,t})\to0$.

Theorem~\ref{thm:GED}, Theorem~\ref{thm:PED}, and
Corollary~\ref{cor:binary} distinguish three claims: a sufficient
condition for every finite menu, a necessary condition tested on
payoff-reversing binary menu restrictions, and an exact iff in the
binary state-matching environment. The arguments use nonanonymous
sampling throughout and allow arbitrary nonstationarity.

\subsection{Anonymous Sampling}
We next ask whether the learning argument survives when the identities of displayed
predecessors are suppressed. Under \emph{anonymous observation}, agent $t$ sees
only the numbers of displayed actions of each type,
\[
N_{a,t}:=\sum_{i<t}D_{i,t}\mathbf 1\{A_i=a\},\qquad a\in A,
\]
rather than the identity-indexed vector of display outcomes. We maintain fresh,
conditionally independent sampling across predecessors but make an additional restrictive assumption that the
selection rule is \emph{stationary}: for every $i<t$,
\[
\mathbb P(D_{i,t}=1\mid A_{<t},\omega)=Q(A_i),
\qquad Q(a)=q_a\in[0,1].
\]
Agents know their calendar date, the selection rule, and the observation
technology, and correctly interpret the resulting social information. 

The central observation is that stationary anonymous sampling admits an
aggregate version of inverse probability weighting. Write
\[
p_{t-1}(a):=\frac{1}{t-1}\sum_{i<t}\mathbf1\{A_i=a\}
\]
for the empirical distribution of predecessors' actions. Whenever $q_a>0$,
\[
\widetilde p_t(a):=\frac{N_{a,t}}{q_a(t-1)}
\]
is the Horvitz--Thompson estimator of $p_{t-1}(a)$ (Horvitz and Thompson 1952; Robins, Rotnitzky, and Zhao 1994). Conditional on the realized
action history, it is unbiased, and
\[
\operatorname{Var}(\widetilde p_t(a)\mid A_{<t})
=\frac{p_{t-1}(a)(1-q_a)}{(t-1)q_a}.
\]
Although individual inclusion probabilities would ordinarily be needed to
implement inverse probability weighting, identities are unnecessary here:
all predecessors taking the same action have the same inclusion probability,
so their contributions can be aggregated into a single count.

If exactly one action $a_0$ is never displayed, its empirical frequency can
still be reconstructed from the adding-up restriction,
\[
\widetilde p_t(a_0)
:=1-\sum_{a\ne a_0}\widetilde p_t(a).
\]
With no such action, all coordinates are estimated by inverse probability
weighting. Projecting the resulting vector onto the simplex $\Delta(A)$ gives
an anonymous-measurable estimate $\widehat p_t$ satisfying, uniformly over
realized predecessor action histories,
\begin{equation}\label{eq:anon-ipw-rate}
\mathbb E\bigl[\|\widehat p_t-p_{t-1}\|_1\mid A_{<t}\bigr]
=O(t^{-1/2}),
\end{equation}
provided at most one action has zero display probability.

This statistical recovery has a direct decision-theoretic implication. An
agent who chooses an action at random according to $\widehat p_t$ approximately
replicates the action of a uniformly sampled predecessor. Thus anonymity
prevents literal imitation of an identified predecessor, but does not
prevent \emph{approximate imitation of the predecessor population}. The
resulting welfare-inheritance inequality can be combined with the value
of a fresh private signal in the same way as under nonanonymous
observation.

\begin{theorem}[Adequate learning under stationary anonymous sampling]
\label{thm:anonymous-stationary}
Suppose that $\Omega$ and $A$ are finite, payoffs are bounded, and private
signals are conditionally i.i.d. and independent of the fresh sampling draws
given the state. Agents are Bayesian and correctly calibrated about the
stationary anonymous selection technology. Assume
\begin{equation}\label{eq:anon-support}
\#\{a\in A:q_a=0\}\le1,
\end{equation}
Then, in every Bayesian equilibrium and under every prior $\pi$,
\[
\lim_{t\to\infty}\mathbb E[u(A_t,\omega)]
=\sum_{\omega\in\Omega}\pi(\omega)\max_{a\in A}u(a,\omega).
\]
If each state has a unique optimal action $a^*(\omega)$, then, for every
$\omega$ with $\pi(\omega)>0$,
\[
\mathbb P_\omega\bigl(A_t=a^*(\omega)\bigr)\longrightarrow1.
\]
\end{theorem}

Theorem~\ref{thm:anonymous-stationary}
therefore extends the original binary anonymous selective-sampling case
$Q(0)=0<Q(1)$ to general finite action sets and payoff functions. In the
binary case, a single positive display probability suffices even if the
other action is never displayed.

The stationarity assumption facilitates the recovery of empirical action frequencies under anonymous observation. When selection probabilities are identical across predecessors, action-dependent selection bias can be corrected using inverse-probability weighting without knowing whose actions were observed. Under nonstationary sampling, however, the observed frequency of an action reflects both how frequently that action was taken and how likely the predecessors who took it were to be observed. These two effects are confounded: a high observed frequency may indicate that an action was chosen by many predecessors, or that it was chosen by relatively few predecessors with high inclusion probabilities. Since predecessor identities are hidden, the observer cannot determine which selection probabilities correspond to the observed actions and therefore cannot generally disentangle these effects through inverse-probability weighting. Our frequency-recovery argument thus does not directly extend to nonstationary sampling. This does not preclude alternative recovery procedures or asymptotic learning under more general sampling technologies, but we restrict the anonymous extension to the stationary case.

\section{Extension: Fragility to Misspecification}

We consider the binary state-matching environment, except that
agents form Bayesian beliefs using a perceived sampling kernel
$\widehat Q$, while observations are generated by the true kernel $Q$.

\begin{proposition}[Fragility to Kernel Misspecification]
\label{prop:fragility}
For every $\varepsilon>0$, there exist nonanonymous sampling kernels
$Q$ and $\widehat Q$ such that
\[
\|Q-\widehat Q\|_\infty<\varepsilon,
\]
and both kernels satisfy ED, so each guarantees asymptotic learning
under correct specification. Nevertheless, when agents use
$\widehat Q$ while the true kernel is $Q$, asymptotic learning fails.
\end{proposition}

The mechanism behind fragility is quite direct. Under the true sampling kernel, a predecessor may choose action \(1\) but fail to be displayed. Under the perceived kernel, however, the absence of a displayed action is interpreted as conclusive evidence that the predecessor chose \(0\). By Bayesian action calibration, this inferred choice leads the current agent to assign a posterior probability of at most \(1/2\) to state \(1\), before incorporating her private signal. Consequently, even when the true state is \(1\), the current agent chooses the wrong action with positive probability. Since this misinterpretation recurs with a probability bounded away from zero at every date, asymptotic learning fails, despite both sampling kernels supporting learning under correct specification.

\section*{References}
\begingroup
\setlength{\parindent}{0pt}
\setlength{\parskip}{0.55em}
\noindent\hangindent=2em\hangafter=1 Acemoglu, Daron, Munther A. Dahleh, Ilan Lobel, and Asuman Ozdaglar (2011). ``Bayesian Learning in Social Networks.'' \emph{Review of Economic Studies} 78(4), 1201--1236.\par

\noindent\hangindent=2em\hangafter=1 Acemoglu, Daron, Ali Makhdoumi, Azarakhsh Malekian, and Asuman Ozdaglar (2022). ``Learning from Reviews: The Selection Effect and the Speed of Learning.'' \emph{Econometrica} 90(6), 2857--2899.\par

\noindent\hangindent=2em\hangafter=1 Bohren, J. Aislinn, and Daniel N. Hauser (2021). ``Learning with Heterogeneous Misspecified Models: Characterization and Robustness.'' \emph{Econometrica} 89(6), 3025--3077.\par

\noindent\hangindent=2em\hangafter=1 Cremin, John W. E. (2024). ``Too Much Information \& The Death of Consensus.'' Working paper, September 12, 2024. \url{https://saet.uiowa.edu/wp-content/uploads/sites/18/gravity_forms/12-854b2acdbe79c39c3d4cf95b7240cb0d/2025/02/John_Cremin.pdf}.\par

\noindent\hangindent=2em\hangafter=1 Frick, Mira, Ryota Iijima, and Yuhta Ishii (2020). ``Misinterpreting Others and the Fragility of Social Learning.'' \emph{Econometrica} 88(6), 2281--2328.\par

\noindent\hangindent=2em\hangafter=1 Frick, Mira, Ryota Iijima, and Yuhta Ishii (2023). ``Belief Convergence under Misspecified Learning: A Martingale Approach.'' \emph{Review of Economic Studies} 90(2), 781--814.\par

\noindent\hangindent=2em\hangafter=1 Fudenberg, Drew, Giacomo Lanzani, and Philipp Strack (2024). ``Selective-Memory Equilibrium.'' \emph{Journal of Political Economy} 132(12), 3978--4020.\par

\noindent\hangindent=2em\hangafter=1 Herrera, Helios, and Johannes H{\"o}rner (2013). ``Biased Social Learning.'' \emph{Games and Economic Behavior} 80, 131--146.\par

\noindent\hangindent=2em\hangafter=1 Horvitz, D. G., and D. J. Thompson (1952). ``A Generalization of Sampling Without Replacement from a Finite Universe.'' \emph{Journal of the American Statistical Association} 47(260), 663--685.\par

\noindent\hangindent=2em\hangafter=1 Hu, Nan, Paul A. Pavlou, and Jie Zhang (2017). ``On Self-Selection Biases in Online Product Reviews.'' \emph{MIS Quarterly} 41(2), 449--471.\par

\noindent\hangindent=2em\hangafter=1 Kartik, Navin, SangMok Lee, Tianhao Liu, and Daniel Rappoport (2024). ``Beyond Unbounded Beliefs: How Preferences and Information Interplay in Social Learning.'' \emph{Econometrica} 92(4), 1033--1062.\par

\noindent\hangindent=2em\hangafter=1 Li, Xinxin, and Lorin M. Hitt (2008). ``Self-Selection and Information Role of Online Product Reviews.'' \emph{Information Systems Research} 19(4), 456--474.\par

\noindent\hangindent=2em\hangafter=1 Liu, Tianhao (2025). ``Misleading via Misspecification: Nonuniformity of Learning Robustness.'' Working paper, June 2025 revision.\par

\noindent\hangindent=2em\hangafter=1 Lobel, Ilan, and Evan D. Sadler (2015). ``Information Diffusion in Networks through Social Learning.'' \emph{Theoretical Economics} 10(3), 807--851.\par

\noindent\hangindent=2em\hangafter=1 Robins, James M., Andrea Rotnitzky, and Lue Ping Zhao (1994). ``Estimation of Regression Coefficients When Some Regressors Are Not Always Observed.'' \emph{Journal of the American Statistical Association} 89(427), 846--866.\par

\noindent\hangindent=2em\hangafter=1 Smith, Lones, and Peter N. S{\o}rensen (2000). ``Pathological Outcomes of Observational Learning.'' \emph{Econometrica} 68(2), 371--398.\par

\noindent\hangindent=2em\hangafter=1 Smith, Lones, and Peter N. S{\o}rensen (2020). ``Rational Social Learning with Random Sampling.'' Working paper, March 1, 2020. \url{https://www.lonessmith.com/wp-content/uploads/2020/03/Rational2020.pdf}.\par

\endgroup

\appendix
\section{Proofs for nonanonymous sampling}
\subsection{Proof of uniform corrected imitation}\label{app:imitation}

We prove Lemma~\ref{lem:replica} directly. All bounds in this section are conditional on an arbitrary complete action history. Such conditioning is understood through the announced observation kernel, even for histories with zero equilibrium probability.

\begin{proof}[Proof of Lemma~\ref{lem:replica}]
\emph{A late window implied by GED.}
Write $m=|B|$ and $r_{i,t}=\rho_{i,t}(B)$. GED states that
\[
\forall K<\infty,\qquad \prod_{i=K+1}^{t-1}(1-r_{i,t})\longrightarrow0.
\]
For each integer $j\ge1$, choose $T_j$ recursively with $T_j>T_{j-1}$, $T_j>j+2$, and
\[
\prod_{i=j+1}^{t-1}(1-r_{i,t})\le\frac1j\qquad\text{for every }t\ge T_j.
\]
For sufficiently large $t$ let $K_t=\max\{j:T_j\le t\}$, and define the finitely many earlier terms arbitrarily. Then $K_t\to\infty$, $K_t<t-1$, and
\[
P_t:=\prod_{i=K_t+1}^{t-1}(1-r_{i,t})\le\frac1{K_t}\longrightarrow0.
\]
Set
\[
R_t=\sum_{i=K_t+1}^{t-1}r_{i,t},\qquad
M_t=\max_{K_t<i<t}r_{i,t},
\]
with the convention $R_t^{-1/2}=+\infty$ when $R_t=0$. We claim that
\begin{equation}\label{eq:two-regime}
\min\{R_t^{-1/2},1-M_t\}\longrightarrow0.
\end{equation}
Otherwise, for some $\eta>0$ and some subsequence, both $R_t\le\eta^{-2}$ and $M_t\le1-\eta$. Using $\log(1-x)\ge-x/(1-x)$ for $x\in[0,1)$, on this subsequence we have
\[
\log P_t=\sum_{i>K_t}\log(1-r_{i,t})
\ge-\sum_{i>K_t}\frac{r_{i,t}}{1-r_{i,t}}
\ge-\frac{R_t}{1-M_t}\ge-\eta^{-3},
\]
contradicting $P_t\to0$. Thus \eqref{eq:two-regime} holds.

\emph{Corrected imitation on the late window.}
At a fixed date $t$, abbreviate $r_i=r_{i,t}$, $R=R_t$, $M=M_t$, and choose $b_i=b_{i,t}\in\arg\min_{a\in B}Q_{i,t}(a)$ with deterministic tie-breaking. By \eqref{eq:second},
\begin{equation}\label{eq:baseline}
Q_{i,t}(a)\ge r_i\qquad\text{for all }a\ne b_i.
\end{equation}

\emph{A corrected vector statistic.} Let $e_a\in\mathbb R^{m}$ be the standard unit vector corresponding to action $a$. For $r_i>0$, introduce an independent thinning seed $V_{i,t}\sim\operatorname{Unif}[0,1]$, also independent of all original model primitives. When $X_{i,t}=a\ne b_i$, accept the displayed action if
\[
V_{i,t}\le\frac{r_i}{Q_{i,t}(a)}.
\]
The ratio is well-defined and at most one by \eqref{eq:baseline}. Let $H_{i,t}$ be the indicator of acceptance. Define the vector $Z_{i,t}$ by
\begin{equation}\label{eq:Z}
Z_{i,t}=\begin{cases}
r_i e_{b_i}+ (e_a-e_{b_i}),&X_{i,t}=a\ne b_i\text{ and }H_{i,t}=1,\\
r_i e_{b_i},&\text{otherwise},
\end{cases}
\qquad(r_i>0).
\end{equation}
When $r_i=0$, define $Z_{i,t}=0$. This statistic is observable using $X_{i,t}$ and the fresh thinning seed. It never uses the true state or a hidden action.

Fix any complete action history $h=(a_1,\ldots,a_{t-1})\in B^{t-1}$. If $a_i=b_i$, then $Z_{i,t}=r_i e_{a_i}$ surely. If $a_i=a\ne b_i$, the action is displayed and accepted with probability $Q_{i,t}(a)[r_i/Q_{i,t}(a)]=r_i$, and $Z_{i,t}=r_i e_{b_i}+H_{i,t}(e_a-e_{b_i})$. Therefore for every such history,
\begin{equation}\label{eq:Zmean}
\EE[Z_{i,t}\mid A_{<t}=h]=r_i e_{a_i}.
\end{equation}
For $a_i\ne b_i$ the centered statistic equals $(H_{i,t}-r_i)(e_{a_i}-e_{b_i})$, whence
\begin{equation}\label{eq:Zvar}
\EE\bigl[\|Z_{i,t}-r_i e_{a_i}\|_2^2\mid A_{<t}=h\bigr]
=2r_i(1-r_i)\le2r_i.
\end{equation}
For $a_i=b_i$ or $r_i=0$ the variance is zero. Fresh independence of the sampling and thinning seeds implies that the $Z_{i,t}$ are independent across predecessors conditional on $A_{<t}=h$.

\emph{Diffuse-exposure regime.} Suppose $R>0$, and define the deterministic weights
\[
w_{i,t}=\frac{r_i}{R}\one\{i>K_t\},
\qquad p^*(h)=\sum_{i>K_t}\frac{r_i}{R}e_{a_i}\in\Delta(B).
\]
Let
\begin{equation}\label{eq:estimate}
\widetilde p_t=\frac1R\sum_{i>K_t}Z_{i,t},
\qquad
\widehat p_t=\Proj_{\Delta(B)}(\widetilde p_t),
\end{equation}
where $\Proj_{\Delta(B)}$ denotes Euclidean projection onto the probability simplex. This is measurable because projection onto a nonempty closed convex subset of Euclidean space is continuous. By \eqref{eq:Zmean}, $\EE[\widetilde p_t\mid h]=p^*(h)$. By conditional independence and \eqref{eq:Zvar},
\begin{equation}\label{eq:ms}
\EE[\|\widetilde p_t-p^*(h)\|_2^2\mid h]
=\frac1{R^2}\sum_{i>K_t}
\EE[\|Z_{i,t}-r_i e_{a_i}\|_2^2\mid h]
\le\frac2R.
\end{equation}
Projection is nonexpansive and $p^*(h)\in\Delta(B)$, so the same upper bound holds with $\widetilde p_t$ replaced by $\widehat p_t$. Conditional on $\widehat p_t$, draw $d_t\in B$ according to probability vector $\widehat p_t$, using an additional independent seed. Then
\[
\cL(d_t\mid h)=\EE[\widehat p_t\mid h],
\]
interpreting both sides as vectors in $\Delta(B)$. Using total variation for finite-action distributions, Jensen's inequality, $\|x\|_1\le\sqrt m\|x\|_2$, and \eqref{eq:ms},
\begin{align}
\left\|\cL(d_t\mid h)-\sum_iw_{i,t}\delta_{a_i}\right\|_{\TV}
&=\frac12\|\EE[\widehat p_t\mid h]-p^*(h)\|_1\notag\\
&\le\frac12\EE[\|\widehat p_t-p^*(h)\|_1\mid h]\notag\\
&\le\frac{\sqrt m}{2}\sqrt{\EE[\|\widehat p_t-p^*(h)\|_2^2\mid h]}
\le\sqrt{\frac{m}{2R}}.\label{eq:weakerr}
\end{align}
Crucially, the bound is uniform over the complete history $h$.

\emph{Nearly perfect predecessor regime.} Choose $i_t>K_t$ such that $r_{i_t}=M$, set $w_{i_t,t}=1$ and the other weights to zero, and define
\begin{equation}\label{eq:singledecision}
d_t=\begin{cases}
X_{i_t,t},&X_{i_t,t}\ne\varnothing,\\
b_{i_t},&X_{i_t,t}=\varnothing.
\end{cases}
\end{equation}
If $a_{i_t}=b_{i_t}$, the rule is exactly correct. If $a_{i_t}\ne b_{i_t}$, the copying error occurs only upon non-display and, by \eqref{eq:baseline},
\begin{equation}\label{eq:perfecterr}
\PP(d_t\ne a_{i_t}\mid h)=1-Q_{i_t,t}(a_{i_t})\le1-M.
\end{equation}
Thus
\[
\left\|\cL(d_t\mid h)-\delta_{a_{i_t}}\right\|_{\TV}\le1-M
\]
for every $h$.

\emph{Selection of the regime and late-predecessor weights.} At each sufficiently large date, compare the two deterministic error bounds in \eqref{eq:weakerr} and \eqref{eq:perfecterr}, adopting the diffuse-exposure construction if $R>0$ and $\sqrt{m/(2R)}\le1-M$, and otherwise adopting the nearly perfect predecessor construction. If $R=0$, only the latter is used. In both cases the support of $w_{\cdot,t}$ lies strictly above $K_t$, and the resulting error satisfies
\begin{equation}\label{eq:epsbound}
\sup_h\left\|\cL(d_t\mid h)-\sum_iw_{i,t}\delta_{a_i}\right\|_{\TV}
\le\eps_t:=\min\left\{\sqrt{\frac{m}{2R_t}},1-M_t\right\}\to0,
\end{equation}
with the first term set to $+\infty$ when $R_t=0$. The convergence follows from \eqref{eq:two-regime}, since $m$ is fixed. Also $\sum_{i\le K}w_{i,t}=0$ for every fixed $K$ once $K_t>K$. This proves both \eqref{eq:escape} and \eqref{eq:replica}.
\end{proof}

\subsection{Welfare inheritance and compactness}\label{app:welfare}

The compactness argument in this section follows Kartik et al(2024), who work with Bayes-plausible distributions of social beliefs and continuity of the private-signal improvement value. We reproduce the needed finite-state formulation to make clear how our nonanonymous corrected-imitation bound replaces the welfare-inheritance step in their setting.

\begin{proof}[Proof of Proposition~\ref{prop:WI}]
For each $t$, take the weights $w_{i,t}$ and social-only randomized rule $d_t$ from Lemma~\ref{lem:replica}, with the uniform total-variation error $\eps_t$ from \eqref{eq:epsbound}. The social-only Bayes value $B_t$ dominates the payoff from any such feasible rule, so $B_t\ge\EE[u(d_t,\omega)]$.

Conditional on a fixed history $h\in B^{t-1}$ and a fixed state $\omega$, the current sampling and auxiliary randomizations have the same distribution as when only $h$ is fixed. In particular, the conditional distribution of $d_t$ does not depend on $\omega$ once $h$ is given. Since any two distributions on $B$ within total variation $\eps_t$ have expected values of $u(\cdot,\omega)$ that differ by at most $C_B\eps_t$,
\[
\EE[u(d_t,\omega)\mid A_{<t}=h,\omega=\omega]
\ge\sum_{i<t}w_{i,t}u(a_i,\omega)-C_B\eps_t.
\]
Integrating over histories and states yields
\[
\EE[u(d_t,\omega)]\ge\sum_{i<t}w_{i,t}\EE[u(A_i,\omega)]-C_B\eps_t
=\sum_{i<t}w_{i,t}U_i^B-C_B\eps_t.
\]
Hence \eqref{eq:WI} follows. Let $L=\liminf_iU_i^B$. For every $\xi>0$, choose a fixed $K$ such that $U_i^B\ge L-\xi$ whenever $i>K$. Since the $U_i^B$ are uniformly bounded and $\sum_{i\le K}w_{i,t}\to0$, the right-hand side of \eqref{eq:WI} has liminf at least $L-\xi$. Let $\xi\downarrow0$ to obtain \eqref{eq:liminf}.
\end{proof}

\begin{lemma}[Continuity of the value functions]\label{lem:continuity}
For finite $\Omega,B$ and bounded $u$, the functions $v_B$, $v_B^F$, $v_B^*$, $I_B$, and $G_B$ are continuous on $\Delta(\Omega)$.
\end{lemma}
\begin{proof}
The functions $v_B$ and $v_B^*$ are respectively the maximum of finitely many linear functions and a linear function of $\mu$, hence are continuous. For each signal $s$, the integrand
\[
\mu\longmapsto\max_{a\in B}\sum_\omega\mu(\omega)f_\omega(s)u(a,\omega)
\]
is continuous. Its absolute value is bounded by $\|u\|_\infty\sum_\omega f_\omega(s)$, which is $\nu$-integrable. Dominated convergence implies continuity of $v_B^F$. Differences preserve continuity.
\end{proof}

\begin{lemma}[Zero gains and common optimal actions]\label{lem:zero}
For any belief $\mu$, $I_B(\mu)=0$ if and only if at least one action optimal under $\mu$ remains optimal after $\mu$-almost every private signal. Likewise $G_B(\mu)=0$ if and only if some $\mu$-optimal action maximizes $u(\cdot,\omega)$ at every $\omega\in\supp\mu$.
\end{lemma}
\begin{proof}
Choose $a_0\in\arg\max_{a\in B}\sum_\omega\mu(\omega)u(a,\omega)$. The expected payoff from continuing with $a_0$ after the signal equals $v_B(\mu)$. The difference between the ex-post maximum and the ex-post payoff of $a_0$ is pointwise nonnegative and integrates to $I_B(\mu)$. If $I_B(\mu)=0$, the difference must vanish for almost every signal, so $a_0$ remains optimal almost surely. Conversely, the existence of one such $a_0$ implies zero incremental signal value. The same argument with full revelation of $\omega$ gives the second statement.
\end{proof}

\begin{proof}[Proof of Lemma~\ref{lem:gain}]
Following the compactness argument in Kartik el al (2024) again , fix the prior $\pi$ and let
\[
\Phi_\pi=\left\{\phi\in\cP(\Delta(\Omega)):\int\mu\,\phi(d\mu)=\pi\right\}.
\]
Since $\Delta(\Omega)$ is compact, $\cP(\Delta(\Omega))$ is weakly compact, and the barycenter condition is closed; hence $\Phi_\pi$ is compact. By Lemma~\ref{lem:continuity}, the maps $\phi\mapsto\int v_B\,d\phi$ and $\phi\mapsto\int I_B\,d\phi$ are continuous.

Let $\Phi_{\pi,\delta}=\{\phi\in\Phi_\pi:\int v_B\,d\phi\le U_B^*-\delta\}$ and suppose this set is nonempty. It is compact. If the minimum of $\int I_B\,d\phi$ over $\Phi_{\pi,\delta}$ were zero, a minimizing $\phi$ would have $I_B(\mu)=0$ for $\phi$-almost every $\mu$ because $I_B\ge0$. By $\mathrm{AK}(B)$, $G_B(\mu)=0$ almost surely under $\phi$. Therefore
\[
\int v_B(\mu)\,\phi(d\mu)=\int v_B^*(\mu)\,\phi(d\mu)
=\sum_\omega\pi(\omega)\max_{a\in B}u(a,\omega)=U_B^*,
\]
contradicting the definition of $\Phi_{\pi,\delta}$. Hence the minimum is strictly positive and defines $\eta(\delta)$.
\end{proof}

\begin{proof}[Proof of Theorem~\ref{thm:GED}]
By Proposition~\ref{prop:WI}, $\liminf B_t\ge L:=\liminf U_t^B$. If $L<U_B^*$, take $\delta=(U_B^*-L)/2>0$. If the antecedent set in Lemma~\ref{lem:gain} is empty, then $B_t>U_B^*-\delta=L+\delta$ at all dates, contradicting the definition of $L$. Otherwise let $\eta=\eta(\delta)>0$. For large enough $t$, $B_t\ge L-\eta/2$. If $B_t>U_B^*-\delta=L+\delta$, then $U_t^B\ge B_t>L+\delta$. Otherwise, Lemma~\ref{lem:gain} gives $J_t\ge\eta$, and therefore $U_t^B=B_t+J_t\ge L+\eta/2$. Either possibility implies $U_t^B\ge L+\min\{\delta,\eta/2\}$ for all sufficiently large $t$, contrary to the definition of $L$. Thus $L=U_B^*$; since $U_t^B\le U_B^*$, the result follows.
\end{proof}

\subsection{Overlap proof of PED necessity}\label{app:necessity}

\begin{lemma}[A common finite-prefix component]\label{lem:prefix}
Under the assumptions of Theorem~\ref{thm:PED}, let $H_K=(A_1,\ldots,A_K)$ in the restricted game and write $P_a^K=\cL(H_K\mid\omega=\omega_a)$ and $P_b^K=\cL(H_K\mid\omega=\omega_b)$. Then
\begin{equation}\label{eq:prefix}
\sum_{h\in\{a,b\}^K}\min\{P_a^K(h),P_b^K(h)\}\ge\beta_{a,b}^{K}.
\end{equation}
For $K=0$ the empty history has probability one and the right side is defined to be one.
\end{lemma}
\begin{proof}
Let $\lambda=F_{\omega_a}\wedge F_{\omega_b}$ be the maximal common subprobability measure of the two signal laws. It has total mass $\beta_{a,b}>0$. The two state-conditioned laws of $(S_1,\ldots,S_K)$ both dominate $\lambda^{\otimes K}$, whose mass is $\beta_{a,b}^{K}$. Enlarge the primitive vector to contain all sampling uniforms and behavioral tie-breaking seeds used through date $K$. Those variables have the same state-independent distribution under $\omega_a$ and $\omega_b$, so both state-conditioned joint primitive laws dominate a common finite measure of mass $\beta_{a,b}^K$. In a fixed equilibrium, $H_K$ is the image of the primitive vector under the same measurable action-generating map in either state. The image of the common measure is a subprobability measure dominated by $P_a^K$ and $P_b^K$, with mass $\beta_{a,b}^K$. Since $H_K$ takes values in a finite set, its maximal common mass is the left side of \eqref{eq:prefix}, which is therefore at least $\beta_{a,b}^K$.
\end{proof}

\begin{lemma}[Pairwise experimental overlap]\label{lem:experiment}
Suppose \eqref{eq:failPED} holds. Let $Z_t=(O_t,S_t)$ denote the complete information of agent $t$ before choosing her action. Then the prior-weighted laws of $Z_{t_j}$ under $\omega_a$ and $\omega_b$ have a common subprobability component of mass at least
\begin{equation}\label{eq:jointmass}
\alpha\,\min\{\pi(\omega_a),\pi(\omega_b)\}\,\beta_{a,b}^{K+1}
\end{equation}
for every $j$.
\end{lemma}
\begin{proof}
Let $E_t=\{X_{i,t}=\varnothing\text{ for every }i>K\}$, and write $O_t^K=(X_{1,t},\ldots,X_{K,t})$. Under the restricted game, $A_i\in\{a,b\}$, so by \eqref{eq:failPED}, for every complete history $h_{<t_j}\in\{a,b\}^{t_j-1}$,
\begin{equation}\label{eq:miss}
\PP(E_{t_j}\mid A_{<t_j}=h_{<t_j})
=\prod_{i=K+1}^{t_j-1}(1-Q_{i,t_j}(h_i))\ge\alpha.
\end{equation}
For $h\in\{a,b\}^K$, let $M_t(y\mid h)$ be the distribution of the date-$t$ prefix display record $O_t^K=y$ conditional on $H_K=h$, induced by the first $K$ fresh sampling draws. This Markov kernel does not depend on the state. Conditional on the full action history, the prefix sampling draws are independent of the tail sampling draws, while the law of the prefix draws depends only on $h$. Therefore for $\omega\in\{\omega_a,\omega_b\}$,
\begin{align}
\pi(\omega)\PP_\omega(O_{t_j}^K=y,E_{t_j})
&=\pi(\omega)\EE_\omega\!\left[M_{t_j}(y\mid H_K)\,\PP(E_{t_j}\mid A_{<t_j})\right]\notag\\
&\ge\alpha\sum_{h\in\{a,b\}^K}\pi(\omega)P_\omega^K(h)M_{t_j}(y\mid h).
\label{eq:lowerprefix}
\end{align}
Put
\[
c_h=\min\{\pi(\omega_a)P_a^K(h),\pi(\omega_b)P_b^K(h)\},
\qquad \gamma_K=\sum_h c_h.
\]
By Lemma~\ref{lem:prefix},
\begin{equation}\label{eq:gamma}
\gamma_K\ge\min\{\pi(\omega_a),\pi(\omega_b)\}\beta_{a,b}^{K}.
\end{equation}
Equation \eqref{eq:lowerprefix} shows that, on the event $E_{t_j}$, the two prior-weighted distributions of the entire social record both dominate the finite measure assigning mass
\[
\alpha\sum_h c_hM_{t_j}(y\mid h)
\]
to each record $(y,\varnothing,\ldots,\varnothing)$. Its total mass is $\alpha\gamma_K$, because $\sum_y M_{t_j}(y\mid h)=1$.

Conditional on the state, the current signal $S_{t_j}$ is independent of $O_{t_j}$ and has law $F_\omega$. Tensoring the common social-record component with the common private-signal component $F_{\omega_a}\wedge F_{\omega_b}$ yields a measure dominated by both prior-weighted joint distributions of $(O_{t_j},S_{t_j})$. Its total mass is $\alpha\gamma_K\beta_{a,b}$, at least the quantity in \eqref{eq:jointmass}. This proves the claim.
\end{proof}

\begin{proof}[Proof of Theorem~\ref{thm:PED}]
Fix an arbitrary Bayesian equilibrium with restricted menu $B=\{a,b\}$. Write $\mathsf P_a$ and $\mathsf P_b$ for the prior-weighted laws of the information $Z_{t_j}$ under $\omega_a$ and $\omega_b$; these finite measures have total masses $\pi(\omega_a)$ and $\pi(\omega_b)$. Let $\lambda_j$ be the common component constructed in Lemma~\ref{lem:experiment}, with mass at least \eqref{eq:jointmass}.

A possibly randomized decision rule is a measurable probability $x(z)\in[0,1]$ of choosing $a$ upon information realization $z$. Its weighted two-state classification error is
\[
\int(1-x(z))\,\mathsf P_a(dz)+\int x(z)\,\mathsf P_b(dz).
\]
Since both $\mathsf P_a$ and $\mathsf P_b$ dominate $\lambda_j$, this error is at least
\[
\int\bigl[(1-x(z))+x(z)\bigr]\lambda_j(dz)
=\lambda_j(\mathsf Z)
\ge\alpha\min\{\pi(\omega_a),\pi(\omega_b)\}\beta_{a,b}^{K+1}.
\]
By strict payoff reversal \eqref{eq:reversal}, every wrong binary decision in state $\omega_a$ or $\omega_b$ incurs full-information regret at least $\Delta_{a,b}$. Every other state's regret is nonnegative because $U_B^*$ uses statewise maximization over $B$. Multiplying the previous bound by $\Delta_{a,b}$ establishes \eqref{eq:loss}. The lower bound holds for every equilibrium and along every subsequence satisfying \eqref{eq:failPED}, so adequate learning is impossible if PED fails.
\end{proof}

\subsection{Binary private signals and corollary}\label{app:binary}

\begin{lemma}[Two-sided unbounded beliefs imply adequate knowledge]\label{lem:binaryAK}
In the binary state-matching model of Section~\ref{sec:binary}, condition \eqref{eq:UB} implies $I_{\{0,1\}}(p)>0$ for every $p\in(0,1)$ and $G_{\{0,1\}}(0)=G_{\{0,1\}}(1)=0$. Hence $\mathrm{AK}(\{0,1\})$ holds.
\end{lemma}
\begin{proof}
Let $p=\PP(\omega=1)\in(0,1)$, and let $L=dF_1/dF_0$. Without seeing the private signal, the optimal action maximizes $p$ and $1-p$. After signal $s$, the posterior odds are $[p/(1-p)]L(s)$. Both events $\{L>(1-p)/p\}$ and $\{L<(1-p)/p\}$ have positive probability under the signal mixture $pF_1+(1-p)F_0$, by the two-sided unboundedness assumption and $F_0\sim F_1$. Thus the optimal action depends nontrivially on the signal. In particular, whichever constant action is optimal before the signal, it is strictly suboptimal on a positive-probability event after the signal. The value of observing the signal is therefore strictly positive, $I(p)>0$. At $p=0$ and $p=1$, the state is already known and $G(p)=0$. It follows that $I(p)=0\Rightarrow G(p)=0$ on the entire interval $[0,1]$.
\end{proof}

\begin{proof}[Proof of Corollary~\ref{cor:binary}]
Two-sided unbounded beliefs ensure that at every $p\in(0,1)$ a fresh private signal has strictly positive decision value, while beliefs $p\in\{0,1\}$ have zero full-information gap. Hence $\mathrm{AK}(\{0,1\})$ holds, and Theorem~\ref{thm:GED} establishes (ii)$\Rightarrow$(i). Conversely, $F_0\sim F_1$ implies $\beta:=\int\min\{dF_0,dF_1\}>0$, and the state-matching payoffs strictly reverse between states zero and one. Theorem~\ref{thm:PED} proves that failure of ED prevents learning in any equilibrium. The identification of expected utility with the probability of a correct action completes the equivalence. Details of the private-signal argument appear in Appendix~\ref{app:binary}.
\end{proof}

\section{Proofs for stationary anonymous sampling}
\label{app:anonymous-stationary}

Throughout this appendix, fix a Bayesian equilibrium and its induced joint
law. Let $m:=|A|$ and $n:=t-1$. Write
\[
p_{t-1}(a)=\frac1n\sum_{i<t}\mathbf1\{A_i=a\},
\qquad
\mathcal A_+:=\{a\in A:q_a>0\}.
\]
All expectations conditional on $A_{<t}$ below are taken with respect to the
fresh date-$t$ display draws and any stated auxiliary randomization.
The selection draws are conditionally independent across predecessors and
are independent of the state and current private signal given the realized
action history.

\begin{lemma}[Uniform anonymous frequency recovery]
\label{lem:anon-frequency-recovery}
Under \eqref{eq:anon-support}, there exists an $\Delta(A)$-valued estimator
$\widehat p_t$, measurable with respect to the anonymous counts
$(N_{a,t})_{a\in A}$, such that for every $t\ge2$ and every action history
$a_{<t}\in A^{t-1}$,
\begin{equation}\label{eq:anon-uniform-bound}
\mathbb E\!\left[
 \|\widehat p_t-p_{t-1}\|_1\mid A_{<t}=a_{<t}
\right]
\le
\sqrt{\frac{2m}{t-1}
 \sum_{a\in\mathcal A_+}\frac{1-q_a}{q_a}}.
\end{equation}
The same bound holds conditional on $A_{<t}$ and the state $\omega$.
\end{lemma}

\begin{proof}
For every $a\in\mathcal A_+$, set
\[
\widetilde p_t(a)=\frac{N_{a,t}}{nq_a}.
\]
If the unique zero-probability action $a_0$ exists, set
$\widetilde p_t(a_0)=1-\sum_{a\in\mathcal A_+}\widetilde p_t(a)$.
Otherwise, $\widetilde p_t$ consists of the preceding coordinates only.
In both cases $\widetilde p_t$ need not belong to $\Delta(A)$; define
\[
\widehat p_t:=\Pi_{\Delta(A)}(\widetilde p_t),
\]
where $\Pi_{\Delta(A)}$ is Euclidean projection onto the probability simplex.

Fix a complete predecessor action history. The count $N_{a,t}$ is a sum of
$n p_{t-1}(a)$ independent Bernoulli$(q_a)$ random variables. Hence, for
$a\in\mathcal A_+$,
\begin{align*}
\mathbb E[\widetilde p_t(a)\mid A_{<t}]&=p_{t-1}(a),\\
\operatorname{Var}(\widetilde p_t(a)\mid A_{<t})
&=\frac{p_{t-1}(a)(1-q_a)}{nq_a}.
\end{align*}
Counts corresponding to different actions involve disjoint sets of
conditionally independent display draws; consequently, the estimators
$\widetilde p_t(a)$ for $a\in\mathcal A_+$ are conditionally independent.
If $a_0$ exists, the adding-up restriction implies
\[
\widetilde p_t(a_0)-p_{t-1}(a_0)
=-\sum_{a\in\mathcal A_+}
  (\widetilde p_t(a)-p_{t-1}(a)).
\]
Therefore, with $z=1$ if a zero-probability action exists and $z=0$
otherwise,
\begin{align*}
\mathbb E\bigl[\|\widetilde p_t-p_{t-1}\|_2^2\mid A_{<t}\bigr]
&=(1+z)\sum_{a\in\mathcal A_+}
   \operatorname{Var}(\widetilde p_t(a)\mid A_{<t})\\
&\le\frac2n\sum_{a\in\mathcal A_+}\frac{1-q_a}{q_a}.
\end{align*}
Because $p_{t-1}\in\Delta(A)$, nonexpansiveness of Euclidean projection
gives
\[
\|\widehat p_t-p_{t-1}\|_2
\le\|\widetilde p_t-p_{t-1}\|_2.
\]
Finally, Jensen's inequality and $\|x\|_1\le\sqrt m\|x\|_2$ yield
\[
\mathbb E\bigl[\|\widehat p_t-p_{t-1}\|_1\mid A_{<t}\bigr]
\le\sqrt{\frac{2m}{n}
 \sum_{a\in\mathcal A_+}\frac{1-q_a}{q_a}},
\]
which is \eqref{eq:anon-uniform-bound}. The entire conditional law of the
date-$t$ display draws given $A_{<t}$ is state-neutral, so conditioning
additionally on $\omega$ leaves the bound unchanged.
\end{proof}

\begin{lemma}[Anonymous approximate imitation and welfare inheritance]
\label{lem:anon-imitation}
Let
\[
U_i:=\mathbb E[u(A_i,\omega)],\qquad
\mu_t:=\mathbb P(\omega\in\cdot\mid(N_{a,t})_{a\in A}),
\qquad B_t:=\mathbb E[v(\mu_t)].
\]
Define the payoff range
\[
C:=\max_{\omega\in\Omega}
 \Bigl(\max_{a\in A}u(a,\omega)-\min_{a\in A}u(a,\omega)\Bigr).
\]
Under the hypotheses of Lemma~\ref{lem:anon-frequency-recovery},
\begin{equation}\label{eq:anon-wi}
B_t\ge\frac1{t-1}\sum_{i<t}U_i-\varepsilon_t,
\qquad
\varepsilon_t:=\frac C2
\sqrt{\frac{2m}{t-1}
 \sum_{a\in\mathcal A_+}\frac{1-q_a}{q_a}}
\xrightarrow[t\to\infty]{}0.
\end{equation}
\end{lemma}

\begin{proof}
Consider an auxiliary randomized decision $d_t$ that, conditional on the
observed anonymous counts, chooses action $a$ with probability
$\widehat p_t(a)$. Its randomization seed is independent of all economic
primitives. This rule uses only information available to agent $t$, and
its expected payoff cannot exceed the optimal payoff based on the anonymous
social observation alone. Thus
\[
B_t\ge\mathbb E[u(d_t,\omega)].
\]
For any two distributions $p,r\in\Delta(A)$ and any state $\omega$,
\begin{equation}\label{eq:anon-tv-payoff}
\left|\sum_{a\in A}(p(a)-r(a))u(a,\omega)\right|
\le \frac C2\|p-r\|_1.
\end{equation}
Fix the realized predecessor action history and the state. Using
\eqref{eq:anon-tv-payoff} with $p=\widehat p_t$ and $r=p_{t-1}$,
and averaging over the fresh sampling and auxiliary decision randomness,
\begin{align*}
\mathbb E[u(d_t,\omega)\mid A_{<t},\omega]
&\ge\sum_{a\in A}p_{t-1}(a)u(a,\omega)
  -\frac C2\mathbb E[\|\widehat p_t-p_{t-1}\|_1
                  \mid A_{<t},\omega]\\
&\ge\frac1{t-1}\sum_{i<t}u(A_i,\omega)-\varepsilon_t,
\end{align*}
where the last inequality follows from
Lemma~\ref{lem:anon-frequency-recovery}. Taking unconditional expectations
proves \eqref{eq:anon-wi}. Importantly, the comparison takes place within
the original equilibrium; the auxiliary rule is only a feasible benchmark
for agent $t$, not a change in the sampling technology or in earlier agents'
strategies.
\end{proof}

\begin{proof}[Proof of Theorem~\ref{thm:anonymous-stationary}]
Fix an arbitrary prior and equilibrium. With the social posterior $\mu_t$
defined in Lemma~\ref{lem:anon-imitation}, let
\[
J_t:=\mathbb E[I(\mu_t)],\qquad
U^*:=\sum_{\omega\in\Omega}\pi(\omega)\max_{a\in A}u(a,\omega).
\]
Because the current private signal is conditionally independent of the
anonymous social observation given the state, Bayesian optimality implies
\begin{equation}\label{eq:anon-payoff-decomp}
U_t=\mathbb E[v^F(\mu_t)]=B_t+J_t.
\end{equation}
In particular, $J_t\ge0$ and $U_t\le U^*$.

Write $L:=\liminf_{t\to\infty}U_t$. Since the predecessor weights
$1/(t-1)$ escape every fixed initial prefix, the Ces\`aro mean of
$(U_i)_{i\ge1}$ has lower limit at least $L$. Combining this with
\eqref{eq:anon-wi} and $\varepsilon_t\to0$ gives
\begin{equation}\label{eq:anon-liminf}
\liminf_{t\to\infty}B_t\ge L.
\end{equation}

We next establish a uniform private-signal improvement whenever social
information leaves a nontrivial payoff gap. The functions $v$, $v^F$,
and $v^*$ are continuous on the compact simplex $\Delta(\Omega)$:
continuity of $v$ and $v^*$ is immediate, while continuity of $v^F$
follows by dominated convergence because payoffs and the state/action sets
are bounded and finite. Consequently, $I$ and $G$ are continuous and
nonnegative. Let
\[
M:=\max_{\mu\in\Delta(\Omega)}G(\mu).
\]
If $M=0$, social information is always adequate and
$B_t=U_t=U^*$ for every $t$, so the conclusion follows. Suppose $M>0$.
For each $\delta\in(0,M]$, the compact set
\[
K_\delta:=\{\mu\in\Delta(\Omega):G(\mu)\ge\delta/2\}
\]
is nonempty. Excludability implies
\[
c_\delta:=\min_{\mu\in K_\delta}I(\mu)>0.
\]
By Bayes plausibility and the linearity of $v^*$,
\[
\mathbb E[G(\mu_t)]=U^*-B_t.
\]
If $B_t\le U^*-\delta$, then
\[
\delta\le\mathbb E[G(\mu_t)]
\le\frac\delta2+M\,\mathbb P(\mu_t\in K_\delta),
\]
so
\begin{equation}\label{eq:anon-uniform-gain}
J_t=\mathbb E[I(\mu_t)]
\ge c_\delta\,\mathbb P(\mu_t\in K_\delta)
\ge\frac{c_\delta\delta}{2M}=: \eta(\delta)>0.
\end{equation}

Suppose toward a contradiction that $L<U^*$ and let
$\delta=(U^*-L)/2>0$. Since $B_t\ge U^*-M$, necessarily
$\delta\le M$. By \eqref{eq:anon-liminf}, for every sufficiently
large $t$,
\[
B_t\ge L-\eta(\delta)/2.
\]
If $B_t>U^*-\delta=L+\delta$, then
$U_t\ge B_t>L+\delta$. Otherwise,
\eqref{eq:anon-uniform-gain} and \eqref{eq:anon-payoff-decomp} imply
\[
U_t\ge B_t+\eta(\delta)\ge L+\eta(\delta)/2.
\]
Thus, for all sufficiently large $t$,
\[
U_t\ge L+\min\{\delta,\eta(\delta)/2\}>L,
\]
contradicting the definition of $L$. Hence $L=U^*$; since
$U_t\le U^*$, we conclude $U_t\to U^*$.

Finally, suppose every state $\omega$ has a unique optimal action
$a^*(\omega)$. Finiteness gives a strictly positive payoff gap
\[
\kappa:=\min_{\omega\in\Omega}\min_{a\ne a^*(\omega)}
 \bigl(u(a^*(\omega),\omega)-u(a,\omega)\bigr)>0.
\]
Therefore
\[
U^*-U_t\ge\kappa\sum_{\omega\in\Omega}\pi(\omega)
 \mathbb P_\omega(A_t\ne a^*(\omega)).
\]
Since the left-hand side converges to zero and the summands on the right
are nonnegative, the probability of a suboptimal action converges to zero
in every state with positive prior probability.
\end{proof}

\section{Proof of fragility to misspecification}
\label{app:fragility}

\begin{proof}
Fix any $\varepsilon>0$ and choose
$\delta\in(0,\min\{\varepsilon,1\})$.
For every $i<t$, define
\[
\begin{aligned}
Q_{i,t}(0)&=0,
& Q_{i,t}(1)&=1-\delta,\\
\widehat Q_{i,t}(0)&=0,
& \widehat Q_{i,t}(1)&=1.
\end{aligned}
\]
Clearly,
\[
\|Q-\widehat Q\|_\infty=\delta<\varepsilon.
\]

We first verify that both kernels satisfy ED.
Under the true kernel,
\[
\max\{Q_{i,t}(0),Q_{i,t}(1)\}=1-\delta.
\]
Thus, for every fixed $K$,
\[
\prod_{i=K+1}^{t-1}
\left(1-\max\{Q_{i,t}(0),Q_{i,t}(1)\}\right)
=\delta^{t-K-1}\longrightarrow0.
\]
Under the perceived kernel, the corresponding
product is zero for every $t\ge K+2$.
Hence both kernels satisfy ED and, by
Corollary~\ref{cor:binary}, each guarantees
asymptotic learning under correct specification.

Now fix any Bayesian equilibrium under the
perceived kernel $\widehat Q$, and implement
its decision rules under the true kernel $Q$.
We denote probabilities in the subjective
model by $\widehat{\mathbb P}$.

Under $\widehat Q$, action $1$ is always
displayed, whereas action $0$ is never
displayed. Consequently, the observation
record reveals the entire predecessor
action history. The subjective equilibrium
is therefore equivalent to a full-history
Bayesian equilibrium.

We first establish a property of this
subjective equilibrium. For every finite
action history $h=(a_1,\ldots,a_j)$
with $a_j=0$,
\begin{equation}
\label{eq:fragility-calibration}
\widehat{\mathbb P}
(\omega=1\mid A_{\le j}=h)
\le\frac12.
\end{equation}

To see this, let
\[
p_j=
\widehat{\mathbb P}
(\omega=1\mid A_{<j},S_j)
\]
be agent $j$'s posterior after observing
the full predecessor history and her
private signal. By Bayesian optimality,
$A_j=0$ implies $p_j\le1/2$, up to
null events. This remains true under
arbitrary optimal tie-breaking.
Therefore, by iterated expectations,
\[
\begin{aligned}
&\widehat{\mathbb P}
(\omega=1\mid A_{\le j}=h)\\
&\qquad=
\widehat{\mathbb E}
[p_j\mid A_{\le j}=h]
\le\frac12.
\end{aligned}
\]
The two-sided unboundedness of private
beliefs and the full-support prior ensure
that every finite binary action history
has positive probability under both
states. Hence the conditional posterior
above is well defined for every $h$.

Consider now an arbitrary agent $t\ge2$.
Let
\[
E_t=\{D_{t-1,t}=0\}
\]
be the event that her immediate
predecessor's action is not displayed.

Under the perceived kernel, $E_t$ implies
that agent $t-1$ chose action $0$.
The entire observation record is
interpreted as a complete action history
ending in $0$. By
\eqref{eq:fragility-calibration}, the
subjective social posterior satisfies
\[
\widehat{\mathbb P}
(\omega=1\mid O_t)\le\frac12
\qquad\text{on }E_t.
\]

Let $L=dF_1/dF_0$ denote the private-signal
likelihood ratio. After incorporating
her private signal, agent $t$'s subjective
posterior odds satisfy
\[
\frac{
\widehat{\mathbb P}(\omega=1\mid O_t,S_t)
}{
\widehat{\mathbb P}(\omega=0\mid O_t,S_t)
}
\le L(S_t)
\qquad\text{on }E_t.
\]
Consequently,
\[
E_t\cap\{L(S_t)<1\}
\subseteq\{A_t=0\},
\]
since action $0$ is then the unique
subjective best response.

Under the true kernel, conditional on
the predecessor's action,
\[
\mathbb P(D_{t-1,t}=0\mid A_{t-1})
=
\begin{cases}
1,&A_{t-1}=0,\\
\delta,&A_{t-1}=1.
\end{cases}
\]
It follows that, in state $1$,
\[
\begin{aligned}
\mathbb P_1(E_t)
&=\mathbb P_1(A_{t-1}=0)
  +\delta\mathbb P_1(A_{t-1}=1)\\
&\ge\delta.
\end{aligned}
\]

Since the current private signal is
independent of predecessor actions and
fresh sampling draws conditional on the
state, we obtain
\[
\begin{aligned}
\mathbb P_1(A_t=0)
&\ge
\mathbb P_1(E_t,L(S_t)<1)\\
&=\mathbb P_1(E_t)
  F_1(\{s:L(s)<1\})\\
&\ge
\delta F_1(\{s:L(s)<1\}).
\end{aligned}
\]
By mutual absolute continuity and the
two-sided unboundedness assumption,
\[
F_1(\{s:L(s)<1\})>0.
\]
Therefore,
\[\inf_{t\ge2}\mathbb P_1(A_t=0)
\ge
\delta F_1(\{s:L(s)<1\})>0.
\]

Thus the probability of choosing the
incorrect action does not converge to
zero in state $1$. Asymptotic learning
fails under misspecification, even though
both $Q$ and $\widehat Q$ satisfy ED.
\end{proof}

\end{document}